\documentclass[11pt,letterpaper]{article}

\usepackage{amsmath,amssymb,amsthm}
\usepackage{algorithmic}
\usepackage[ruled,vlined,linesnumbered]{algorithm2e}
\usepackage{xspace, units}
\usepackage{textcomp}
\usepackage{cite}
\usepackage{xcolor}
\usepackage{tikz,typearea}

\newtheorem{theorem}{Theorem}

\newcommand{\classNP}{{\sf NP}}

\newcommand{\RR}{\ensuremath{\mathbb{R}}}

\newcommand{\e}{{\mathrm e}}

\newcommand{\calL}{\ensuremath{\mathcal{L}}}
\newcommand{\calR}{\ensuremath{\mathcal{R}}}
\newcommand{\calC}{\ensuremath{\mathcal{C}}}
\newcommand{\calN}{\ensuremath{\mathcal{N}}}
\newcommand{\calA}{\ensuremath{\mathcal{A}}}
\newcommand{\calW}{\ensuremath{\mathcal{W}}}

\newcommand{\Ex}[1]{\mbox{\rm\bf E}\left[#1\right]}

\newcommand{\junk}[1]{}

\newcommand{\OPT}{\ensuremath{\mathrm{OPT}}}

\DeclareMathOperator*{\STAT}{STAT}
\DeclareMathOperator*{\FIXED}{FIXED}
\DeclareMathOperator*{\SECPRICE}{\mathit{SEC-PRICE}}
\newcommand{\growingmid}{\mathrel{}\middle|\mathrel{}}

\usepackage{vmargin}
\setpapersize{USletter}
\setmarginsrb{1in}{1in}{1in}{1in}{0mm}{0mm}{7mm}{7mm}

\usepackage{todonotes}

\begin{document}

\title{Universally Truthful Secondary Spectrum Auctions}

\author{Martin Hoefer%
	\thanks{Max-Planck-Institut f\"ur Informatik and Saarland University, Saarbr\"ucken, Germany, {\tt 
	mhoefer@cs.rwth-aachen.de}. Supported by DFG through Cluster of Excellence 
	MMCI.}
	\and Thomas Kesselheim%
	\thanks{Dept.\ of Computer Science, Cornell University, Ithaca, NY, USA. {\tt kesselheim@cs.cornell.edu}. Supported by a fellowship within the Postdoc-Programme of the German Academic Exchange Service (DAAD) and by DFG through UMIC Research Center at RWTH Aachen University.} 
}
\date{} 

\maketitle

\begin{abstract}
%\bloodbath
%We treat algorithmic and economic aspects of secondary spectrum assignment in wireless networks. 
We present algorithms for implementing local spectrum redistribution in wireless networks using a mechanism design approach. For example, in single-hop request scheduling, secondary users are modeled as rational agents that have private utility when getting assigned a channel for successful transmission. We present a rather simple algorithmic technique that allows to turn existing and future approximation algorithms and heuristics into truthful mechanisms for a large variety of networking problems. In contrast to previous work, our approach works for virtually all known interference models in the literature, including the physical model of interference based on SINR. It allows to address single-hop and multi-hop scheduling, routing, and even more general assignment and allocation problems. Our mechanisms are randomized and represent the first universally-truthful mechanisms for these problems with rigorous worst-case guarantees on the solution quality. 
%They are truthful independently of the randomization, random choices are used only to provide small bounds on approximation factors. 
In this way, our mechanisms can be used to obtain guaranteed solution quality even with risk-averse or risk-seeking bidders, for which existing approaches fail. 
% Finally, our bounds can be used to provide guarantees for social welfare of the solution and revenue to the auctioneer.
\end{abstract}

\thispagestyle{empty}
\setcounter{page}{0}
\clearpage 

\section{Introduction}

%Wireless networks are nowadays ubiquitous, and their successful operation crucially relies on feasible transmission in the presence of interference and noise. 
A major challenge in the development of future wireless networking technology lies in spectrum scarcity, i.e., the limited availability of frequency bands for existing and new devices or services. To a large extent this problem results from the static allocation and licensing rules currently in use, where licenses for frequency bands are given to large service providers for entire countries. There is a major research effort underway in computer science and engineering to overcome this static allocation approach. The main idea is to motivate the licensed \emph{primary users} to open up their bands in local areas and enable \emph{secondary users} to use these local spectrum opportunities.

An attractive approach to implement secondary usage are \emph{auctions}. There is a rich theory in economics showing that auctions allow to implement global objectives in a market with rational participants and private information. Auction mechanisms allow to implement secondary spectrum usage as a market, in which primary users can sell access to their unused licensed spectrum bands to secondary users~\cite{Berry10}. 
% Secondary users have a private valuation for getting a channel, and they must pay some amount of money at the end of the auction. They act economically as rational selfish agents and consider their \emph{utility} as valuation minus payments. The auction \emph{mechanism} is an algorithm that takes the reported valuations as input, allocates channels and charges payments to secondary users. The goal of a mechanism is usually to allocate channels to maximize the \emph{social welfare}, i.e., the total ``happiness'' as measured by the sum of (reported) valuations of users. 
%
In this scenario, secondary users are assumed to be selfish and rational. They can manipulate the outcome of a mechanism by misreporting their valuation and try to obtain a desired channel using smaller payments. Therefore, we desire a \emph{truthful} mechanism that computes allocation and payments in a way that no user has an incentive to lie. It ensures that every user maximizes his utility by reporting the true valuation. For decades, the most prominent technique to design truthful mechanisms was VCG~\cite{Vickrey61,Clarke71,Groves73}, which applies only if the chosen allocation optimizes social welfare exactly. Unfortunately, for non-trivial spectrum auctions social welfare maximization is \classNP-hard. The challenge is thus to find mechanisms that (1) are truthful, (2) run in polynomial time, and (3) provide allocations with near-optimal social welfare.

\paragraph{Our Contribution}
We introduce a unified framework to design simple polynomial-time mechanisms for a large variety of problems within secondary spectrum usage and prove non-trivial worst-case guarantees on their social welfare. We heavily extend the current literature on secondary spectrum auctions in several ways. 

First, many existing secondary spectrum auctions model interference as a binary property using, e.g., unit-disk graphs, interference boundaries, or are otherwise based on local binary conflicts~\cite{Zhou08,Zhou09,Gopinathan11,Zhu12}. These approaches lack many important features of realistic signal propagation. Instead, we here use realistic interference models based on the signal-to-interference-plus-noise ratio (SINR), for which allocation problems are often significantly more complicated to analyze. 

Second, while there exist approaches to secondary spectrum auctions with SINR~\cite{Hoefer12,HoeferKV11}, they are mostly unsuitable for practical use. Being polynomial-time algorithms in principle, their main drawback is the time-intensive use of randomized meta-rounding~\cite{LaviS11,Carr02} and the ellipsoid method for convex optimization. The meta-rounding technique relies on %decomposing a fractional allocation by solving an exponential-size linear program with 
the ellipsoid method and iterative application of separation oracles. This approach is used by virtually all known secondary spectrum auctions with non-trivial performance guarantees, even the ones for binary interference models~\cite{Gopinathan11,Zhu12}. In contrast, our mechanisms are fast, surprisingly simple to grasp and implement, without complicated convex optimization techniques. Nevertheless, we prove non-trivial worst-case bounds on their allocation quality.

Third, the meta-rounding approach used in related work yields randomized mechanisms that are \emph{truthful in expectation}, i.e., if users care only about the \emph{expectation} of utility, they have no incentive to lie. Truthfulness in expectation is a strong assumption, because users must be risk-neutral. Such mechanisms lose their truthfulness if users are, e.g., risk-averse or risk-seeking or take further parameters of the utility distribution into account. In contrast, our mechanisms are \emph{universally truthful}, randomization has no effect on incentives. Our mechanisms can be thought of as probability distributions over deterministic truthful mechanisms. We use randomization only to guarantee good social welfare.

Fourth, while we present our approach for secondary markets with single-hop transmissions in the SINR model, it turns out to be much more flexible. In fact, it is applicable to a large variety of mechanism design problems in (wireless) networks. The main criterion is that feasibility in the underlying allocation problem is downward closed. This holds, e.g., for all allocation problems that can be modeled via packing linear or integer programs. We highlight this extension of our approach for multi-hop routing problems with binary~\cite{Zhu12} or even SINR interference.

%Finally, some of our results continue to hold when we measure quality of an outcome in terms of \emph{revenue}, an issue that is only very recently starting to be explored~\cite{GopinathanRev11}.

The general idea of our approach is as follows. Similar to~\cite{Dobzinski12}, we first determine a ``representative'' valuation of the users by independently at random moving each user into a statistics group with small probability. The users in this group are asked their valuation and discarded without channel and payments. Based on their reports, a representative valuation is determined, from which we obtain a random take-it-or-leave-it price $p$. Each surviving user is asked if he would want to buy a channel at price $p$. If not, the user leaves the auction, otherwise he survives. For the set of surviving users, we then determine an allocation using an arbitrary polynomial-time approximation algorithm for the channel allocation problem. Finally, each user that ends up with a channel pays the price $p$, all others pay nothing. While truthfulness of this scheme can be shown rather directly, the challenge is proving that this general framework allows to obtain non-trivial worst-case guarantees on social welfare and revenue.

\paragraph{Related Work}
%\label{sec:related}
In recent years, there have been numerous studies of different flavors of the problem of maximizing the number of successful transmissions in a wireless network with interference. Early works model this problem essentially as a maximum independent set problem in graphs~\cite{ErlebachJS05, Nieberg08}. However, as physical conditions are only poorly captured this way, starting with a paper by Moscibroda and Wattenhofer~\cite{MoscibrodaW06} attention has shifted to models based on SINR. Many approximation algorithms use fixed power assignments, e.g., uniform ones~\cite{AndrewsD09, HalldorssonW09}, or, more generally, monotone ones~\cite{HalldorssonM11,HalldorssonHMW13}. Better results in terms of the approximation factor can, however, be achieved by making also the power assignment itself subject to optimization~\cite{Kesselheim11,Kesselheim12}. To illustrate our techniques, we use the latter algorithm, which is the best algorithm to solve the combined selection and power-control problem so far. However, we also show that our approach is much more general, allowing it to be used with any of the algorithms mentioned above as well as any future algorithm.

Designing auction mechanisms for secondary spectrum markets has attracted increased research interest in recent years. Usually, socially optimal channel assignment poses \classNP-hard graph coloring and maximum independent set problems. In these cases, the classic approach of VCG for designing truthful mechanisms~\cite{Vickrey61,Clarke71,Groves73} cannot be applied efficiently. Most approaches to spectrum auctions are single-parameter problems, where each user has a single numerical value for getting any one of a set of channels. They are often based on the classic monotonicity characterization of truthful mechanisms~\cite{Myerson81}. Zhou et al~\cite{Zhou08} studied monotone deterministic algorithms for a secondary spectrum auction in a graph-based interference model. This was extended by Zhou and Zheng~\cite{Zhou09} to a double auction framework. More recently, Gopinathan et al~\cite{Gopinathan11} studied truthful-in-expectation mechanisms via randomized meta-rounding that allow to bound the worst-case approximation factor and allow to incorporate additional objectives like fairness. In a related work, Gopinathan et al~\cite{GopinathanRev11} also studied revenue optimization using concepts introduced for revenue-optimizing auctions of digital goods. Very recently, Zhu et al~\cite{Zhu12} proposed deterministic monotone mechanisms for a single-parameter multi-hop routing model. They also provided truthful-in-expectation mechanisms based on randomized meta-rounding with provable worst-case performance.

All the above mentioned works address single-parameter domains. The general case, in which each user can have a different valuation for each subset of acquired channels, was studied in our work~\cite{HoeferKV11} for single-hop scheduling in a general framework encompassing graph-based and SINR interference models. Very recently, we provided improved results in the same framework for the popular subclasses of symmetric and submodular user valuations~\cite{Hoefer12}. In both papers, we proposed truthful-in-expectation mechanisms based on randomized meta-rounding with provable worst-case guarantees. Our guarantees depend on the inductive independence number~\cite{AkcogluADK02}, a non-standard graph parameter that can capture feasibility in virtually all known interference models.

\paragraph{Outline}
The remainder of this paper is structured as follows. We define the initial scenario and mechanism design basics in Section~\ref{sec:prelim}. In Section~\ref{sec:SINR} we describe our mechanism for single-hop request scheduling in the SINR model. In Section~\ref{sec:general}, we then show how to generalize our technique to many other problems, including single- and multi-hop scheduling in the SINR and other interference models. Finally, we conclude in Section~\ref{sec:conclude}.

\section{Preliminaries}
\label{sec:prelim}

\paragraph{Network Model}

We will first explain our approach in the context of single-hop wireless transmission scheduling and then generalize our technique later on. We assume there is a primary user that holds a set $\calC$ of $k$ channels in a local area. There is a set $\calN$ of $n$ secondary users that strive to obtain a channel. Each user corresponds to a communication request or \emph{link} between points in a metric space. Link $i$ consists of sender $s_i$ and receiver $r_i$. 

For each channel, the primary user assigns licenses to allow a subset $\calL \subseteq \calN$ of users to use the channel. If link $i$ gets a license, sender $s_i$ transmits on the channel a signal with power $\sigma_i$. We will consider different power assignments below. The signal strength decays exponentially over distance with exponent $\alpha > 1$. Given the subset $\calL \subseteq \calN$ of links and a power assignment $\sigma : \calL \to \RR_{> 0}$, the transmission on link $i \in \calL$ is successful if the signal-to-interference-plus-noise ratio (SINR) $\gamma_i(\calL,p)$ is above some predefined threshold $\beta > 0$:
\begin{equation}
\label{eq:SINR}
\gamma_i(\calL,\sigma) \quad = \quad \frac{ \nicefrac{\sigma_i}{d_{ii}^\alpha}}{N + \displaystyle \sum_{\substack{j \in \calL, j\neq i}} \nicefrac{\sigma_j}{d_{ji}^\alpha}} \quad \ge \quad \beta \enspace,
\end{equation}
where $d_{ji}$ is the distance between sender $s_j$ and receiver $r_i$ and $N$ is ambient noise. This condition captures the intuition that at $r_i$ the decayed signal of $s_i$ is significantly louder than interference from decayed signals of other senders and noise.

Each user $i$ has a benefit $w_i$ for being able to transmit successfully on a channel. In particular, its benefit is $w_i > 0$ if and only if he gets a license for at least one channel on which condition~\eqref{eq:SINR} is fulfilled. Otherwise, its benefit is 0. The goal of the primary user is to compute an allocation $\calA = (\calL_1,\ldots,\calL_k)$ of secondary users in $\calN$ to channels to maximize \emph{social welfare}, i.e., the sum of benefits of successful users.

\paragraph{Mechanism Design}

Benefits are private information of secondary users. Before the allocation each secondary user $i$ must report its benefit $w_i$ for a successful transmission; we say user $i$ makes \emph{bid} $b_i$. Users are rational and selfish, they would like to obtain a channel at the lowest possible cost. Depending on the channel allocation algorithm, user $i$ might benefit from misreporting its value and have an incentive to lie $b_i \neq w_i$. 

This fundamental problem has been studied for several decades in the area of \emph{mechanism design}. To avoid manipulation and set incentives for truthful bids, we design a \emph{truthful mechanism}. It collects the vector of bids $b = (b_i)_{i \in \calN}$, allocates users to channels, and charges payments $p_i(b)$ to user $i \in \calN$. User $i \in \calN$ has a quasi-linear \emph{utility}
\[ u_i(b) = \begin{cases} w_i - p_i(b) & \text{\eqref{eq:SINR} holds on at least one channel assigned to $i$}\\
0 & \text{otherwise\enspace.}
\end{cases}
\]
A deterministic mechanism is \emph{truthful} if no user profits from lying. More formally, we require that
\begin{equation}
\label{eq:utility}
u_i(b_i,b_{-i}) \le u_i(w_i,b_{-i}) \hspace{0.5cm} \text{for all $b_i \ge 0$ and $b_{-i}$\enspace,}
\end{equation}
where we use $b_{-i}$ to denote the vector $b$ of bids excluding bid $b_i$. The classic technique to compose deterministic truthful mechanisms is VCG~\cite{Vickrey61,Clarke71,Groves73}, but it applies only when the allocation maximizes social welfare exactly. Maximizing social welfare in the channel allocation problem is \classNP-hard even for special cases~\cite{AndrewsD09}. We have to find other ways to obtain truthful mechanisms with good social welfare guarantees that run in polynomial time.

Towards this end, we study randomized mechanisms, which are \emph{truthful in expectation} if \eqref{eq:utility} holds for the \emph{expected} utility of every user. Truthfulness in expectation has the drawback that the outcome of random coins must be invisible to the users, and except for the expectation every user must be indifferent to all other parameters of its utility distribution. 

A much stronger condition is fulfilled if a randomized mechanism is \emph{universally truthful}. We can interpret a universally truthful mechanism as having a number of deterministic truthful mechanisms and initially making one randomized decision which one of the deterministic mechanisms to apply. Here truthfulness is independent of randomization, so the mechanism could make the random decision in public before collecting bids. The reason we use randomization is to obtain a better approximation of social welfare. 

\section{Universally Truthful Auctions for Single-Hop SINR Scheduling}
\label{sec:SINR}
The general structure of our mechanism is presented in Algorithm~\ref{alg:universalframework}. It is inspired by a universally truthful mechanism for combinatorial auctions by Dobzinski et al.~\cite{Dobzinski12}. At the beginning, the mechanism decides at random whether the actual allocation algorithm is executed at all. The purpose of this step is to address the case when there is a super-dominant bidder whose bid is much higher than the sum of valuations of all other bidders. To treat this case, the mechanism runs a second-price auction in which only a single bidder is served with a small, constant probability $\varepsilon$. In this case, only the bidder having the highest bid is allocated a channel and is charged the second highest bid.

With probability $1 - \varepsilon$, the actual allocation algorithm is performed. We consider prior-free mechanisms without any knowledge on the bidders' valuations. Therefore, we sample a number of bidders that are asked for their valuation and discarded afterwards. If the auction mechanism is applied repeatedly in practice, knowledge about typical bidder valuations might become available due to experience, which allows to skip this sampling step. Having obtained knowledge on ``typical'' valuations, we determine a price in a random experiment. The price $p$ is obtained by choosing $X$ uniformly at random from $\{0, 1, \ldots, \lceil \log n \rceil + 1\}$ and setting $p := 2^{-X} \cdot B$, where $B$ is the largest bid that was observed in the sampling phase. Afterwards, we only consider bidders that are willing to pay at least $p$. For these bidders, we compute an allocation, in which each bidder is either selected and assigned to a channel in $\calC$ or it is discarded (see below). In the end, each bidder that is selected and assigned to a channel is charged $p$. The discarded bidders are not charged.

\begin{algorithm}[t]
set $\SECPRICE$ to $1$ with probability $\varepsilon$, otherwise to $0$\;
\eIf{$\SECPRICE = 1$}{
Let $i^\ast$ be a bidder such that $b_{i^\ast} = \max_{i \in \calN} b_i$\;
Pick an arbitrary channel in $\calC$ and allocate it to $i^\ast$\;
Charge $p_{i^\ast}(b) := \max_{i \in \calN \setminus \{ i^\ast\} } b_i$\;
All other bidders $i$ are not allocated and charged $p_i(b) := 0$\;
}{
\ForAll{$i \in \calN$}{add $i$ to $\STAT$ with probability $\varepsilon$, otherwise to $\FIXED$}
Set $B := \max_{i \in \STAT} b_i$\;
choose $X$ uniformly at random from $\{0, 1, \ldots, \lceil \log n \rceil + 1\}$\;
Set $p := 2^{-X} \cdot B$\;
Set $M := \{ i \in \FIXED \mid b_i \geq p\}$\;
Run \texttt{Unweighted\-Packing}$(M, \calC)$ which returns allocation $(\calL_1,\ldots,\calL_k)$\;
If $i \in \bigcup_{j=1}^k \calL_j$, charge $p_i(b) := p$, otherwise $p_i(b) := 0$\;
}
\caption{Universally-truthful framework}
\label{alg:universalframework}
\end{algorithm}

The selection of bidders and allocation to channels is delegated to a procedure called \texttt{Unweighted\-Packing}$(M, \calC)$. Given a set $M \subseteq \calN$ of bidders, it computes a partial allocation of bidders in $M$ to channels in $\calC$ that is feasible with respect to the SINR constraints. Note that the routine does not get to know the valuations of bidders, instead it aims at maximizing only the \emph{number} of selected bidders. Their actual valuations are neglected at this stage. We here assume that the application scenario allows power control, which allows us to employ Algorithm~\ref{alg:packingwithpc} originally presented in~\cite{Kesselheim11}. It is known that this algorithm always computes feasible allocations. In particular, the desired power assignment fulfilling the SINR constraints of all selected links in Line~\ref{line:power} exists. Furthermore, it is a constant-factor approximation algorithm for the optimization problem of maximizing the number of feasible requests on $k$ channels (for a detailed treatment see~\cite{Kesselheim11}). In our case, this implies that for a given set $M$, we have the following property: Let $S = \bigcup_{j=1}^k \calL_j$ be the set of bidders selected on any channel after executing \texttt{Unweighted\-Packing}$(M, \calC)$. Let furthermore $F \subseteq M$ be any arbitrary set of bidders such that there is a feasible allocation of all bidders in $F$ to the channels in $\calC$. Then we have $\lvert S \rvert \geq \kappa \cdot \lvert F \rvert$ for some constant $\kappa \leq 1$.

If the application scenario does not allow power control, we can alternatively implement {\tt Unweigh\-ted\-Packing}$(M,\calC)$ by adapting algorithms for fixed power assignments. %such as the ones presented by Goussevskaia et al.~\cite{Goussevskaia09} and Halld\'orsson and Mitra~\cite{Halldorsson11}.
We discuss this issue as part of the general treatment in the next section. 

\begin{algorithm}[b]
initialize $\calL_1', \ldots, \calL_k' = \emptyset$ \;
\For{$\ell' \in \calR$ in order of increasing length}{
\If{there is some $j \in \calC$ such that $\sum_{\substack{(s, r) \in \calL_j \\ d(s, r) < d(s', r')}} \frac{d(s, r)^\alpha}{d(s, r')^\alpha} + \frac{d(s, r)^\alpha}{d(s', r)^\alpha} \leq \frac{1}{2 \cdot 3^\alpha \cdot \left(4 \beta + 2\right)}$}{
add $\ell'$ to $\calL_j$\;
}
}
compute power assignment $\sigma$ such that all SINR constraints are fulfilled\label{line:power}\;
return $(\calL_1, \ldots, \calL_k)$\;
\caption{Unweighted\-PackingPC$(M, \calC)$ \cite{Kesselheim11}}
\label{alg:packingwithpc}
\end{algorithm}
 
We will first show that Algorithm~\ref{alg:universalframework} is a universally truthful mechanism. Given any fixed outcome of the random experiments, no bidder can profit from misreporting the valuation.

\begin{theorem}
\label{theorem:truthfulsinr}
Algorithm~\ref{alg:universalframework} is universally truthful.
\end{theorem}

\begin{proof}
We have to show that the mechanism is a distribution over (deterministic) truthful mechanisms. Let us now fix the outcome of $\SECPRICE$. In case of $\SECPRICE = 0$, fix furthermore all random choices leading to the division into $\STAT$ and $\FIXED$, and also the random choice of $X$ leading to $p$.

Now, consider a single bidder $i \in \calN$. We have to show that she cannot gain by misreporting the valuation. For this purpose, we distinguish between the outcomes of the randomization.

\textbf{Case 1: \boldmath$\SECPRICE = 1$} In this case, we perform the classic deterministic second-price Vickrey auction for a single item~\cite{Vickrey61} that is known to be monotone~\cite{Myerson81} and truthful.
%The reason is that a bidder $v \neq v^\ast$ could only change the outcome (and therefore her utility) by overbidding $v^\ast$. Under these circumstances, she would be charged $b_{v^\ast}$ by the mechanism, which exceeds her utility and renders the utlity negative. Bidder $v^\ast$ in turn cannot increase the utility either because she is allocated a channel and cannot influence the payment by her bid.

\textbf{Case 2: \boldmath$\SECPRICE = 0$ and $i \in \STAT$} Bidder $i$ will not be allocated anything, independent of the bid. Thus, she has no incentive to misreport.

\textbf{Case 3: \boldmath$\SECPRICE = 0$ and $i \in \FIXED$} The bid $b_i$ only determines whether or not $i$ is contained in set $M$. The price $p$ cannot be influenced by $i$, because it is determined only by the bidders in $\STAT$. Now, let us assume that $i$ reports a false bid $b_i'$, which results in a different allocation $S'$.

The two cases that either $b_i, b_i' \ge p$ or $b_i,b'_i \le p$ are irrelevant as in both cases $M$, the allocation and the payment are the same for $b_i$ and $b'_i$. Hence, it remains to consider the cases $b_i < p \leq b_i'$ and $b_i' < p \leq b_i$. In the first case, $i$ can be contained in $S'$ but not in $S$. If $i \not\in S'$, she has zero utility 0 for both bids. If $i \in S'$, then bidder $i$ has to pay $p$, which is more than her valuation. Hence, her quasi-linear utility is negative as opposed to zero utility for bid $b_i$. In the case $b_i' < p \leq b_i$, we have in turn that $i \in S$ but $i \not\in S'$. In case of bidding $b_i$, the utility is $b_i - p \geq 0$, whereas in case of bidding $b_i'$ it is $0$. In any case, the utility is maximized by bidding $b_i$.
\end{proof}

This shows that our algorithm is a truthful mechanism independent of random choices. We now proceed to the proof that the allocation returned by the mechanism provides a non-trivial worst-case approximation guarantee. Our proof proceeds in two cases. If the bid of the highest bidder is very large compared to the optimal social welfare, then we resort to the single-item Vickrey auction executed with some probability that gives him a channel. In this way, our algorithm secures the high bid with significant probability. Otherwise, we bound the social welfare of our allocation using the total revenue generated through the prices $p$ paid by the allocated bidders. As no bidder gets negative utility, the revenue is a lower bound for the welfare of the allocation obtained by the algorithm. 

\begin{theorem}
\label{theorem:sinrApx}
Let $\calA$ be the allocation returned by the mechanism and $\calA^\ast$ be the allocation optimizing social welfare. Then we have
\[
\Ex{b(\calA)} = \Omega\left( \frac{\varepsilon}{\log n}\right) \cdot b(\calA^\ast) \enspace.
\]
\end{theorem}

\begin{proof}
We denote the highest bid by $B^\ast = \max_{i \in \calN} b_i$. Note that the channels are important for feasibility of the solution but not for the social welfare. In particular, a bidder is indifferent about the channel he obtains. Hence, we will focus on $S$ to denote simply \emph{the set of bidders allocated to at least one channel} by the algorithm. Similarly, we use $S^\ast$ to denote the set of bidders receiving at least one channel in the optimum. 

\textbf{Case 1: $\boldsymbol{b(\calA^\ast) < 8 \cdot B^\ast}$} To bound the social welfare in this case, we only consider the event that $\SECPRICE = 1$. This happens with probability $\varepsilon$. We know that, provided that $\SECPRICE = 1$, the social welfare of the computed solution is precisely $B^\ast$. That is, the expected social welfare is at least 
\[
\Ex{b(\calA)} \ge \varepsilon \cdot B^\ast \geq \frac{\varepsilon}{8} \cdot b(\calA^\ast) \enspace.
\]
This completes the proof for the first case.

\textbf{Case 2: $\boldsymbol{b(\calA^\ast) \geq 8 \cdot B^\ast}$} Let $i^\ast$ be a bidder with the highest bid $b_{i^\ast} = B^\ast$. To bound the expected revenue, we only consider cases, in which $\SECPRICE = 0$ and $i^\ast \in \STAT$.

Under the assumption that $i^\ast \in \STAT$, we denote by $S_j$ the set of bidders getting at least one channel when $X = j$. In particular, $S_j$ is the set of bidders allocated by the packing algorithm when executed on the set $M_j = \{ i \in \FIXED \mid b_i \geq 2^{-j} \cdot B^\ast\}$. Using $p = 2^{-j} \cdot B^\ast$, the expected revenue of the mechanism can be bounded by
\[
\frac{1}{\lceil \log n \rceil + 2} \cdot \Ex{\sum_{j=0}^{\lceil \log n \rceil + 1} 2^{-j} \cdot B^\ast \cdot \lvert S_j \rvert \growingmid i^\ast \in \STAT} \enspace.
\]
By linearity of expectation, this is equal to
\begin{equation}
\frac{1}{\lceil \log n \rceil + 2} \sum_{j=0}^{\lceil \log n \rceil + 1} 2^{-j} \cdot B^\ast \cdot \Ex{\lvert S_j \rvert \growingmid i^\ast \in \STAT} \enspace.
\end{equation}

Furthermore, let us define $S_j^{\ast} = \{ i \in S^{\ast} \mid b_i \geq 2^{-j} \cdot B^\ast \}$ as the set of bidders receiving a channel in the optimum and having a bid at least $2^{-j} \cdot B^\ast$. Obviously, $S_j^{\ast} \cap \FIXED$ is a subset of $M_j$, and it can be turned into a feasible allocation because it allocates only a subset of bidders from the optimum. Hence, using our assumption on the quality of algorithm {\tt Unweighted\-Packing}, we have $\lvert S_j \rvert \geq \kappa \cdot \lvert S_j^{\ast} \cap \FIXED \rvert$ for a suitable constant $\kappa$. Furthermore, we can bound the expected size of $S_j^\ast \cap \FIXED$ by
\begin{align*}
& \Ex{\lvert S_j^\ast \cap \FIXED \rvert \growingmid i^\ast \in \STAT} \quad = \quad \Ex{\lvert S_j^\ast \cap \FIXED \setminus \{i^\ast\} \rvert} \\
= \quad & (1 - \varepsilon) \lvert S_j^\ast \setminus \{i^\ast\} \rvert 
 \quad \geq \quad (1 - \varepsilon)(\lvert S_j^\ast \rvert - 1) \enspace.
\end{align*}

The value of this is at least
\begin{align*}
& \frac{1}{\lceil \log n \rceil + 2} \sum_{j=0}^{\lceil \log n \rceil + 1} 2^{-j} B^\ast \kappa \cdot (1 - \varepsilon)(\lvert S_j^\ast \rvert - 1)
\; \geq \; \frac{(1 - \varepsilon) \kappa}{\lceil \log n \rceil + 2} \left( \left( \sum_{j=0}^{\lceil \log n \rceil + 1} 2^{-j} B^\ast \cdot \lvert S_j^\ast \rvert \right) - 2 B^\ast \right).
\end{align*}

To derive the approximation factor now, we have to bound $b(\calA^\ast)$ in terms of the cardinalities of the $S_j^\ast$ sets. As a first step, we consider links that are not contained in any of these sets. For this purpose, let us define $Y = \calN \setminus S_{\lceil \log n \rceil + 1}$. The contribution of these links to the social welfare is at most
\[
\sum_{i \in Y} b_i \leq n \cdot 2^{-\lceil \log n \rceil -1} \cdot B^\ast \leq \frac{b(\calA^\ast)}{2} \enspace,
\]
because $b(\calA^\ast) \geq B^\ast$. This yields for the remaining links
\[
\sum_{j = 0}^{\lceil \log n \rceil + 1} 2^{-j} \cdot B^\ast \cdot \lvert S_j \rvert \geq \frac{b(\calA^\ast)}{4} \enspace.
\]
Combining this bound with the previously obtained bound on revenue, we get that the expected revenue is at least
\[
\frac{(1 - \varepsilon) \kappa}{\lceil \log n \rceil + 2} \left(\frac{b(\calA^\ast)}{4} - 2 B^\ast \right) \quad \geq \quad \frac{(1 - \varepsilon) \kappa}{\lceil \log n \rceil + 2} \cdot \frac{b(\calA^\ast)}{8} \enspace.
\]
For small $\varepsilon$, this shows that the revenue is in $\Omega\left(\frac{\varepsilon \cdot \kappa}{\log n}\right) \cdot b(\calA^\ast)$. As no bidder has negative utility, this shows the same bound for the social welfare. With $\kappa$ being constant, we obtain the claim of the theorem in Case 2 with $\SECPRICE = 0$ and $i^\ast \in \STAT$.

Finally, $\SECPRICE = 0$ and $i^\ast \in \STAT$ happens with probability $(1-\varepsilon)\cdot\varepsilon$.
In all other cases, we underestimate the welfare of the solution computed by the mechanism by 0. For the expected social welfare obtained by our algorithm this yields
\[
\Ex{b(\calA)} \ge \frac{(1-\varepsilon)^2 \cdot \varepsilon \cdot \kappa}{(\lceil \log n \rceil + 2)\cdot 8} \cdot b(\calA^\ast) \enspace.
\]
This proves the claim for Case 2 and yields the theorem.
\end{proof}

%Note that our proof technique allows to directly extend our results to revenue approximation. In particular, the maximum revenue that a truthful mechanism can obtain is obviously upper bounded by the social welfare of the solution. Thus, if we can bound the revenue in terms of optimal social welfare, the same bound applies to optimal revenue. Thus, the analysis and the bound of Case 2 in the previous theorem hold similarly for revenue. If the highest bid is very large, however, no truthful mechanism can obtain a revenue guarantee in terms of social welfare~\cite{HartlineChapter07}. In this case, we still obtain the second largest bid as revenue which is optimal for single-item auctions. A more stringent analysis using other revenue benchmarks in prior-free mechanism design is left for future work.

\section{General Technique for Universally Truthful Network Auctions}
\label{sec:general}
In this section, we demonstrate that our approach from the previous section allows to build universally truthful mechanisms for virtually all secondary spectrum auctions. To illustrate our results, we will focus on secondary  network auctions studied Zhu et al.~\cite{Zhu12}. In this model, \emph{secondary networks} are selfish users participating in an auction. Each user has a network $G_i$ with a dedicated source node $s_i$ and destination node $d_i$ that are publicly known. User $i$ wants to connect $s_i$ to $d_i$ without interference, and any interference-free $s_i$-$d_i$-path has value $w_i$ for her. The valuation $w_i$ is private information. The allocation problem of the mechanism is to select a subset $S \subseteq \calN$ of the users. For each selected user $i \in S$, it has to allocate a path of links connecting $s_i$ to $d_i$ in $G_i$. For each used link in any of the $G_i$ graphs, it furthermore has to select a channel from $\calC$. The selection of paths and channels has to be free of interference. The interference constraint is modeled using an additional conflict graph $H$ as follows. The vertex set of $H$ consists of all links of the $G_i$-graphs. The edges of $H$ indicate which links in the $G_i$-graphs collide due to interference. An allocation of channels to links and paths to bidders is only feasible if none of the used links that are on the same channel are connected by an edge in $H$.

The overall approach of the mechanism is as before. The users here are again single-parameter agents -- their private information is the single numerical value $w_i$ representing the value for an interference-free $s_i$-$d_i$-path in $G_i$. Zhu et al~\cite{Zhu12} present a deterministic heuristic for the allocation problem that satisfies Myerson's monotonicity condition~\cite{Myerson81} and can therefore be turned into a truthful mechanism. Furthermore, they discuss a truthful-in-expectation mechanism based on randomized meta-rounding with a provable approximation guarantee. The guarantee depends on parameters like the maximum length of an allocated path and structural properties of the conflict graph $H$.

There are further fundamental similarities with our single-hop scheduling approach. This scenario also exhibits a \emph{packing structure}. Each bidder $i$ either obtains some interference-free path, for which her valuation is $w_i$, or she is not served at all. In this sense, an allocation determines a set $W$ of winners that become successful in the solution, we will say they are assigned a ``license'' for successful transmission. The social welfare of the allocation is given by $\sum_{i \in W} w_i$. By packing structure we refer to the property that for any feasible allocation, we can remove any winner $i' \in W$ without changing the feasibility of the allocation for any other bidder. In both, secondary network and single-hop scheduling scenarios, this is due to the fact that by removing bidders, interference for the other bidders can only decrease. In the secondary network setting, removing bidders and their paths will never introduce new conflicts. They same is true for request scheduling, because the SINR can only increase by removing bidders and SINR constraints cannot be violated this way.

To formalize this property, we introduce a family $\calW \subseteq 2^{\calN}$ consisting of sets of bidders. Each $W \in \calW$ represents a subset of bidders that are assigned a license in a particular feasible solution. That is, if we consider a fixed allocation, each $i \in W$ obtains a license, whereas no $i \not\in W$ gets a license. If a scenario has a packing structure, the family $\calW$ is \emph{downward closed}: For any $W \in \calW$ all subsets are contained in $\calW$ as well, i.e., if $W \in \calW$ and $W' \subseteq W$, then also $W' \in \calW$.

For the single-hop scenario considered in Section~\ref{sec:SINR}, $\calW$ contains all sets $W \subseteq \calN$ of links that can be allocated to the channels in $\calC$ such that there is a power assignment fulfilling all SINR constraints. In case of the secondary network setting, it contains all sets $W \subseteq \calN$ for which there is an allocation to channels guaranteeing an interference-free selection of paths.

\begin{algorithm}[t]
For each bidder $i$: If there is no allocation that yields $i$ as winner, remove $i$ from $\calN$\;
Set $\SECPRICE$ to $1$ with probability $\varepsilon$, otherwise to $0$\;
\eIf{$\SECPRICE = 1$}{
Let $i^\ast$ be a bidder such that $b_{i^\ast} = \max_{i \in \calN} b_i$\;
Compute the allocation in which only $i^\ast$ is a winner\;
Charge $p_{i^\ast}(b) := \max_{i \in \calN \setminus \{ i^\ast\} } b_i$\;
All other bidders $i$ are charged $p_i(b) := 0$\;
}{
\ForAll{$i \in \calN$}{add $i$ to $\STAT$ with probability $\varepsilon$, otherwise to $\FIXED$}
Set $B := \max_{i \in \STAT} b_i$\;
Choose $X$ uniformly at random from $\{0, 1, \ldots, \lceil \log n \rceil + 1\}$\;
Set $p := 2^{-X} \cdot B$\;
Set $M := \{ i \in \FIXED \mid b_i \geq p\}$\;
Run \texttt{UnweightedPacking}$(M)$ which returns allocation $\calA$\;
If $i$ contained in $\calA$, charge $p_i(b) := p$, otherwise $p_i(b) := 0$\;
}
\caption{Universally-truthful framework}
\label{alg:universalframeworkpacking}
\end{algorithm}

Based on this property, we can generalize the framework introduced in the previous section as shown in Algorithm~\ref{alg:universalframeworkpacking}. We first adjust the set of bidders $\calN$ to remove all bidders that can never be part of a feasible allocation. Because $\calW$ is downward closed, it suffices to check for each single $i \in \calN$ if there is a feasible allocation with $W = \{i\}$. This step is unnecessary for single-hop scheduling with power contorl in the previous section, as sufficiently large powers guarantee that every single bidder could become successful if she is the only one assigned to transmit. 

In the main routine, our mechanism relies on an algorithm \texttt{UnweightedPacking}$(M)$. Given at set $M$ of candidate bidders, it calculates a feasible allocation $\calA$, in which only the bidders in $W$ are served. Here algorithm \texttt{UnweightedPacking} again neglects the bids. 

\begin{theorem}
Algorithm~\ref{alg:universalframeworkpacking} is universally truthful.
\end{theorem}

Proving truthfulness can be done by literally the same arguments as used in the proof of Theorem~\ref{theorem:truthfulsinr}. Also, we can again show an approximation guarantee based on the applied algorithm \texttt{UnweightedPacking}$(M)$. In the following, we assume that this algorithm is a $\psi$-approximation, meaning that for the winning set $W$, we have $\lvert W \rvert \geq \psi \cdot \max_{W' \in \calW \cap 2^M} \lvert W' \rvert$. %In words, the number of winners in $W$ is at most a $\psi$-factor smaller than in any feasible solution serving only bidders in $M$. 
Given this guarantee, we can show the following approximation factor for the complete mechanism.

\begin{theorem}
\label{theorem:packingApx}
Let $\calA$ be the allocation returned by the mechanism and $\calA^\ast$ be the allocation optimizing social welfare. Then we have
\[
\Ex{b(\calA)} = \Omega\left( \frac{\varepsilon \psi}{\log n}\right) \cdot b(\calA^\ast) \enspace.
\]
\end{theorem}

The proof of this theorem is almost identical to the proof of Theorem~\ref{theorem:sinrApx} (with winner sets $W$ taking the role of sets $S$ of scheduled users) and is presented for completeness in the appendix.

\section{Applications and Results}

In the following, we discuss some applications of our approach, for which we obtain universally truthful mechanisms with non-trivial worst-case guarantees using our general technique.

\paragraph{Power Control}
As seen above, we can design universally truthful mechanisms for single-hop request scheduling in the SINR model with power control that yield approximation factors of $O(\log n)$. If we cannot use power control and must resort to fixed power schemes, there might exist request that do not overcome noise with their fixed power assignment. Here we can use the general framework and implement the allocation algorithm \texttt{UnweightedPacking}$(M)$ by, e.g., adapting the algorithms of Goussevskaia et al.~\cite{Goussevskaia09}, Halld\'orsson and Mitra~\cite{HalldorssonM11}, or Halld\'orsson et al.~\cite{HalldorssonHMW13}. While these algorithms have been designed to maximize the number of feasible requests on a single channel, it is rather easy to observe that they maintain their approximation guarantees even if we can assign requests to more than one channel (for completeness, see Theorem~\ref{theorem:extend} in the Appendix). This way, we obtain universally truthful mechanisms for single-hop scheduling with, e.g.,  approximation factors of $O(\log n)$ for uniform, linear, and square-root power assignments.

\paragraph{Interference Models} 
Instead of SINR interference, we can also use any graph-based interference models, such as the protocol model or the disk-graph model. For the two latter models, there exist constant-factor approximation algorithms to maximize the number of feasible requests on a single channel. For the disc-graph model a 5-approximation algorithm is folklore. For the protocol model, one can construct conflict graphs with constant inductive independence number~\cite{Wan09,HoeferKV11}, using which a simple greedy algorithm constructs a constant-factor approximation for capacity maximization~\cite{AkcogluADK02}. As we prove in Theorem~\ref{theorem:extend} in the Appendix, we can generically extend these algorithms to constant-factor approximation algorithms for capacity maximization with multiple channels. Using these algorithms for \texttt{UnweightedPacking}$(M,\calC)$, we obtain universally truthful mechanisms for single-hop scheduling in disc graph or protocol models with approximation factors of $O(\log n)$.

\paragraph{Secondary Networks}
In the secondary network model above, we can use algorithms described by Zhu et al~\cite{Zhu12} as a subroutine for \texttt{UnweightedPacking}$(M)$ in our general framework. In this way, we obtain a universally truthful mechanism with approximation factor $O(g(L_{\max},\Delta) \cdot \log n)$. For a description of the factor $g(L_{\max},\Delta)$, see~\cite{Zhu12}.
%
%\item It is an interesting open problem to study the secondary network scenario also using SINR interference. Clearly, the SINR constraints yield again a packing structure, and hence the scenario can be addressed using our technique. Designing approximation algorithms with provable worst-case guarantees for this case is an interesting open problem for future work. Any such algorithm could immediately be used to implement \texttt{UnweightedPacking}$(M)$ in our framework and thereby would result in a universally truthful auction.

\paragraph{General Packing Problems}
More generally, \emph{every} allocation problem in (wireless) networks that has a packing structure and allows a non-trivial approximation algorithm can be turned into a universally truthful mechanism using our approach, while spending only an additional $O(\log n)$ factor in the guarantee. This includes, e.g., a large variety of routing, scheduling, congestion, and assignment problems in networks.

\section{Conclusion}
\label{sec:conclude}
In this paper, we have presented a technique to design universally truthful secondary spectrum auctions that extends the literature on existing spectrum auctions in several ways. Furthermore, it allows to turn arbitrary approximation algorithms into truthful mechanisms in a large variety of previously studied models. To apply our technique to further settings, one only needs to adapt the respective allocation algorithm \texttt{UnweightedPacking}$(M)$, which does not have to fulfill any further requirements. In any case, we guarantee that bidders cannot benefit from misreporting their valuation, even if they are risk-averse or risk-seeking.

An interesting topic for future work is to generalize the approach to multi-parameter mechanism design. In terms of secondary spectrum auctions, this means, for example, that each user's valuation depends on the amount of interference that it is exposed to, rather than only considering whether interference is too high or not.

\bibliographystyle{plain}
\bibliography{../../../Bibfiles/literature,../../../Bibfiles/martin}

\clearpage

\appendix
\section{Appendix}
\subsection{Missing Proofs}
\begin{proof}(of Theorem~\ref{theorem:packingApx})
We denote the highest bid by $B^\ast = \max_{i \in \calN} b_i$. Let $W$ denote the set of winning bidders in the computed allocation. Similarly, let $W^\ast$ to denote the set of winning bidders in allocation that maximizes social welfare. 

\textbf{Case 1: $\boldsymbol{b(\calA^\ast) < 8 \cdot B^\ast}$} To bound the social welfare in this case, we only consider the event that $\SECPRICE = 1$. This happens with probability $\varepsilon$. We know that, provided that $\SECPRICE = 1$, the social welfare of the computed solution is precisely $B^\ast$. That is, the expected social welfare is at least 
\[
\Ex{b(\calA)} \ge \varepsilon \cdot B^\ast \geq \frac{\varepsilon}{8} \cdot b(\calA^\ast) \enspace.
\]
This completes the proof for the first case.

\textbf{Case 2: $\boldsymbol{b(\calA^\ast) \geq 8 \cdot B^\ast}$} Let $i^\ast$ be a bidder with the highest bid $b_{i^\ast} = B^\ast$. To bound the expected revenue, we only consider cases, in which $\SECPRICE = 0$ and $i^\ast \in \STAT$.

Under the assumption that $i^\ast \in \STAT$, we denote by $W_j$ the set of winning bidders when $X = j$. That is, $W_j$ is the set of bidders allocated by the packing algorithm when executed on the set $M_j = \{ i \in \FIXED \mid b_i \geq 2^{-j} \cdot B^\ast\}$. Using $p = 2^{-j} \cdot B^\ast$, the expected revenue of the mechanism can be bounded by
\[
\frac{1}{\lceil \log n \rceil + 2} \cdot \Ex{\sum_{j=0}^{\lceil \log n \rceil + 1} 2^{-j} \cdot B^\ast \cdot \lvert W_j \rvert \growingmid i^\ast \in \STAT} \enspace.
\]
By linearity of expectation, this is equal to
\begin{equation}
\frac{1}{\lceil \log n \rceil + 2} \sum_{j=0}^{\lceil \log n \rceil + 1} 2^{-j} \cdot B^\ast \cdot \Ex{\lvert W_j \rvert \growingmid i^\ast \in \STAT} \enspace.
\end{equation}

Furthermore, let us define $W_j^{\ast} = \{ i \in W^{\ast} \mid b_i \geq 2^{-j} \cdot B^\ast \}$ as the set of winning bidders in the optimum and having a bid at least $2^{-j} \cdot B^\ast$. Obviously, $W_j^{\ast} \cap \FIXED$ is a subset of $M_j$. As $\calW$ is downward closed, it is also contained in $\calW$. Hence, using our assumption on the quality of algorithm {\tt UnweightedPacking}, we have $\lvert W_j \rvert \geq \psi \cdot \lvert W_j^{\ast} \cap \FIXED \rvert$. Furthermore, we can bound the expected size of $W_j^\ast \cap \FIXED$ by
\begin{align*}
& \Ex{\lvert W_j^\ast \cap \FIXED \rvert \growingmid i^\ast \in \STAT} \\
& = \Ex{\lvert W_j^\ast \cap \FIXED \setminus \{i^\ast\} \rvert} \\
& = (1 - \varepsilon) \lvert W_j^\ast \setminus \{i^\ast\} \rvert \\
& \geq (1 - \varepsilon)(\lvert W_j^\ast \rvert - 1) \enspace.
\end{align*}

The value of this is at least
\begin{align*}
& \frac{1}{\lceil \log n \rceil + 2} \sum_{j=0}^{\lceil \log n \rceil + 1} 2^{-j} \cdot B^\ast \cdot \psi \cdot (1 - \varepsilon)(\lvert W_j^\ast \rvert - 1) \\
& \geq \frac{(1 - \varepsilon) \psi}{\lceil \log n \rceil + 2} \left( \left( \sum_{j=0}^{\lceil \log n \rceil + 1} 2^{-j} \cdot B^\ast \cdot \lvert W_j^\ast \rvert \right) - 2 B^\ast \right) \enspace.
\end{align*}

To derive the approximation factor now, we have to bound $b(\calA^\ast)$ in terms of the cardinalities of the $W_j^\ast$ sets. As a first step, we consider bidders that are not contained in any of these sets. For this purpose, let us define $Y = \calN \setminus S_{\lceil \log n \rceil + 1}$. The contribution of these bidders to the social welfare is at most
\[
\sum_{i \in Y} b_i \leq n \cdot 2^{-\lceil \log n \rceil -1} \cdot B^\ast \leq \frac{b(\calA^\ast)}{2} \enspace,
\]
because $b(\calA^\ast) \geq B^\ast$. This yields for the remaining links
\[
\sum_{j = 0}^{\lceil \log n \rceil + 1} 2^{-j} \cdot B^\ast \cdot \lvert W_j \rvert \geq \frac{b(\calA^\ast)}{4} \enspace.
\]
Combining this bound with the previously obtained bound on revenue, we get that the expected revenue is at least
\[
\frac{(1 - \varepsilon) \psi}{\lceil \log n \rceil + 2} \left(\frac{b(\calA^\ast)}{4} - 2 B^\ast \right) \geq \frac{(1 - \varepsilon) \psi}{\lceil \log n \rceil + 2} \cdot \frac{b(\calA^\ast)}{8} \enspace.
\]
For small $\varepsilon$, this shows that the revenue is in $\Omega\left(\frac{\varepsilon \cdot \psi}{\log n}\right) \cdot b(\calA^\ast)$. As no bidder has negative utility, this shows the same bound for the social welfare. 

Finally, $\SECPRICE = 0$ and $i^\ast \in \STAT$ happens with probability $(1-\varepsilon)\cdot\varepsilon$.
In all other cases, we underestimate the welfare of the solution computed by the mechanism by 0. For the expected social welfare obtained by our algorithm this yields
\[
\Ex{b(\calA)} \ge \frac{(1-\varepsilon)^2 \cdot \varepsilon \cdot \psi}{(\lceil \log n \rceil + 2)\cdot 8} \cdot b(\calA^\ast) \enspace.
\]
This proves the claim for case 2 and yields the theorem.
\end{proof}

\subsection{Extending Capacity Maximization Algorithms to Multiple Channels}
In order to apply algorithms that are meant to solve the single-channel allocation problem in the context of multiple channels, we can use the following reduction: We iteratively fill each of the sets $\calL_1, \ldots, \calL_k$ by applying the single-channel algorithm on all remaining users. We can show that we only lose a constant-factor in terms of the approximation guarantee.

\begin{theorem}
\label{theorem:extend}
Given an algorithm with approximation factor of $\psi$ for the unweighted single-channel allocation problem, we obtain a $((1 - \nicefrac{1}{\e}) \psi)$-approximation algorithm for the unweighted $k$-channel problem.
\end{theorem}

\begin{proof}
Let $\calL_1, \ldots, \calL_k$ be the set of users selected for the respective channels by the algorithm. Let $\OPT$ be the total number of selected users in the optimal solution.

Let us consider the $j$th run of the algorithm, which selects the users for set $\calL_j$. We know that at this stage at most $\sum_{i=1}^{j-1} \lvert \calL_i \rvert$ have been selected in previous rounds. This means that at least $\OPT - \sum_{i=1}^{j-1} \lvert \calL_i \rvert$ users that are selected in the optimal solution are still unselected. In the optimal allocation these are allocated to at most $k$ different channels, meaning that at least $\frac{\OPT - \sum_{i=1}^{j-1} \lvert \calL_i \rvert}{k}$ can share a channel. By using the approximation guarantee of the single-channel algorithm, we get
\[
\lvert \calL_j \rvert \geq \psi \frac{\OPT - \sum_{i=1}^{j-1} \lvert \calL_i \rvert}{k} \enspace.
\]
Adding $\sum_{i=1}^{j-1} \lvert \calL_i \rvert$ to both sides, we get
\begin{align*}
\sum_{i=1}^j \lvert \calL_i \rvert & \geq \sum_{i=1}^{j-1} \lvert \calL_i \rvert + \psi \frac{\OPT - \sum_{i=1}^{j-1} \lvert \calL_i \rvert}{k} \\
& = \left( 1 - \frac{\psi}{k} \right) \sum_{i=1}^{j-1} \lvert \calL_i \rvert + \frac{\psi}{k} \OPT \enspace.
\end{align*}
Solving this recursion, this yields
\begin{align*}
\sum_{i=1}^k \lvert \calL_i \rvert & \geq \sum_{i=0}^{k-1} \left( 1 - \frac{\psi}{k} \right)^i \cdot \frac{\psi}{k} \OPT \\ 
& = \frac{\left( 1 - \frac{\psi}{k} \right)^k}{1 - \left( 1 - \frac{\psi}{k} \right)} \cdot \frac{\psi}{k} \OPT \\
& = \left( 1 - \frac{\psi}{k} \right)^k \OPT \\
& \geq \left( 1 - \frac{1}{\e} \right) \psi \OPT \enspace.
\end{align*}
This shows the claim.
\end{proof}

\end{document}